\documentclass[11pt]{article}
\usepackage{amsmath,amssymb}

\DeclareMathOperator{\arctanh}{arctanh}

\usepackage{float} 
\newcommand{\be}{\begin{equation}}
	\newcommand{\ee}{\end{equation}}
\newcommand{\ba}{\begin{array}}
	\newcommand{\ea}{\end{array}}
\newcommand{\bea}{\begin{eqnarray*}}
	\newcommand{\eea}{\end{eqnarray*}}
\newcommand{\bean}{\begin{eqnarray}}
	\newcommand{\eean}{\end{eqnarray}}
\newtheorem{theorem}{Theorem}[section]
\newcommand{\bth}{\begin{theorem}}
 
\newtheorem{lemma}{Lemma}[section]
\newcommand{\blem}{\begin{lemma}}
	\newcommand{\elem}{\end{lemma}}
\newtheorem{proposition}{Proposition}[section]
\newcommand{\bprop}{\begin{proposition}}
	\newcommand{\eprop}{\end{proposition}}
\newtheorem{remark}{Remark}[section]
\newcommand{\brem}{\begin{remark}}
	\newcommand{\erem}{\end{remark}}

\newcommand{\bno}{\begin{notation}}
	\newcommand{\eno}{\end{notation}}

\newcommand{\ce}{\text{ce}}
\newcommand{\se}{\text{se}}
\newcommand{\Ce}{\text{Ce}}

\newcommand{\Ke}{\text{Ke}}

\def\CC{\rm \hbox{C\kern-.56em\raise.4ex
		\hbox{$\scriptscriptstyle |$}\kern+0.5 em }}
\def\Box{\leavevmode\vbox{\hrule
		\hbox{\vrule\kern5pt\vbox{\kern5pt}%
			\vrule}\hrule}}

\usepackage{graphicx}
\usepackage{color}
\newcount\clegende \clegende=0
\def\legende#1{\global\advance \clegende by 1
	$$ \hbox{\hss\small Figure \the\clegende~: #1 \hss} $$}

\begin{document}
	\title{Spectral Analysis of a Quantum Waveguide with Elliptical Window.}
	\author{ H. NAJAR \thanks{Department of Mathematics, Faculty of Sciences Monastir, Monastir University, Tunisia. Research Laboratory: Algebra, Geometry and Spectral Theory: LR11ES53 (Email: hatemnajar@yahoo.fr )}
		\and  F. CHOGLE  \thanks{Department of Applied Mathematics and Sciences, Khalifa University of Science and Technology, 127788, Abu Dhabi, United Arab Emirates (Email: firoz.chogle@ku.ac.ae)}}
	\date{}
	
	\maketitle \abstract{We investigate the Dirichlet Laplacian in two spatial waveguides coupled through an elliptic window. The elliptic geometry breaks rotational symmetry and introduces anisotropy through the semi-axes of the aperture, which modifies the coupling of transverse modes and the low-lying spectrum. We prove that the operator has a finite number of discrete eigenvalues below the threshold of the essential spectrum and study their dependence on the geometric parameters of the ellipse. In contrast to the circular case, the elliptic setting gives rise to spectral effects such as eigenvalue splitting. Numerical simulations illustrate the variation of the first eigenvalues and the ground state with the window geometry.}

	\vskip1cm
	
\textit{ This paper is dedicated to the memory of Oleg Olendeski, whose passion for research and generosity greatly influenced us. His contributions to this field were invaluable. His collaboration with H. N.  and his supervision of F. C. have had a positive impact on this work.}
\section{Introduction}

Boundary value problems for elliptic operators occupy a central position in both mathematical physics and the theory of partial differential equations. They arise naturally in the modeling of quantum confinement, wave propagation, and diffusion processes, while at the same time posing fundamental analytical questions related to regularity, spectral asymptotics, and the influence of geometry. In this work, we investigate the spectral properties of the Laplacian subject to mixed Dirichlet–Neumann boundary conditions, a classical setting that lies between purely Dirichlet and purely Neumann problems and exhibits features characteristic of both.

\subsection{Short description of the problem}

Let $\Omega$ be a two- or three-dimensional domain whose boundary is decomposed into two disjoint parts,
\[
\partial\Omega = \partial\Omega_D \cup \partial\Omega_N,
\]
on which homogeneous Dirichlet and Neumann boundary conditions are imposed, respectively. We consider the spectral problem for the Laplacian, namely the determination of eigenvalues $E_n$ and corresponding eigenfunctions $f_n({\bf r})$, $n=1,2,\ldots$, satisfying
\begin{equation}\label{Dirichlet1}
	\left.f_n({\bf r})\right|_{\partial\Omega_D} = 0,
\end{equation}
and
\begin{equation}\label{Neumann1}
	{\bf n}\cdot\nabla f_n({\bf r})\big|_{\partial\Omega_N} = 0,
\end{equation}
where ${\bf n}$ denotes the outward unit normal vector to the boundary. This mixed-boundary-value problem is commonly known as the \emph{Zaremba problem} \cite{Zaremba1}. It interpolates between the purely Dirichlet and purely Neumann cases and thus provides a natural framework for studying intermediate spectral regimes encountered in both physical and mathematical contexts.

A defining difficulty of the Zaremba problem is the behavior of solutions near the interface between the two types of boundary conditions. Denoting
$$
\Sigma = \overline{\partial\Omega_D} \cap \overline{\partial\Omega_N},
$$
one observes that standard elliptic regularity breaks down at $\Sigma$, even when $\Omega$ and the boundary partition are smooth. In the vicinity of this Dirichlet–Neumann junction, eigenfunctions typically exhibit singular behavior characterized by non-integer powers of the distance to the interface. As a result, solutions may fail to belong to $H^2(\Omega)$ and require a refined analytical description using weighted Sobolev spaces or Kondrat’ev-type asymptotic expansions \cite{Gri, Kond}.

Although the Laplacian with mixed boundary conditions has compact resolvent and hence a purely discrete spectrum in bounded domains, the presence of boundary singularities influences spectral asymptotics and the qualitative structure of eigenfunctions. Compared to homogeneous boundary conditions, the analysis of eigenvalue distributions, nodal sets, and geometric–spectral relations becomes more subtle, reflecting the interplay between geometry and boundary heterogeneity.

These analytical challenges motivate the use of numerical methods, which allow one to explore spectral properties in situations where explicit analysis is not possible. However, the numerical treatment of mixed boundary value problems is itself delicate, as it requires the accurate implementation of heterogeneous boundary conditions and the proper resolution of interface singularities, especially in higher-dimensional or unbounded settings.

Mixed Dirichlet–Neumann boundary conditions arise naturally in a variety of applications. Classical examples include heat conduction in partially insulated bodies \cite{Dowker1}, quantum billiards with partially reflective boundaries \cite{Wiersig1}, and wave propagation in systems exhibiting ray-splitting phenomena \cite{Prange1}. A particularly important class of examples is provided by quantum waveguides, where different boundary conditions are prescribed on different parts of the boundary. This setting naturally extends the Zaremba problem to unbounded domains and leads to new spectral phenomena.

\subsection{Quantum waveguides}

Quantum waveguides offer a rich testing ground for studying the interaction between geometry, boundary conditions, and spectral properties of differential operators. They have attracted considerable attention due to their relevance in nanostructures and electromagnetic wave propagation; see, for example, the monograph \cite{hurt} and the references therein.

Pioneering work by Exner and collaborators established many fundamental results on the spectral theory of quantum waveguides under various geometric and boundary configurations \cite{exner1, exner4, exner2, exner3}. Subsequent studies addressed discrete and random models \cite{Stokel, speis1, naj6}, further clarifying the influence of disorder and geometry.

It is now well understood that the spectral properties of waveguides are strongly affected by geometric perturbations. Curvature-induced bound states \cite{bulla1, exner4, exner2} and states generated by window coupling between straight waveguides \cite{exner2, hurt} are well-known examples. Waveguides with purely Neumann boundary conditions have also been investigated \cite{Krejcirik1, naz}. When Dirichlet and Neumann conditions coexist, additional spectral effects arise, including the emergence of bound states in otherwise straight geometries \cite{najrandom}. The influence of magnetic fields on such mixed-boundary configurations has also been studied \cite{Borisov13, Olendski2, Olendski3}. For the electric field influence see \cite{najpop}.

In two-dimensional straight Dirichlet waveguides, it was shown that the introduction of a Neumann window produces at least one bound state below the essential spectrum \cite{bulla1, exner2}. Such configurations model coupled semiconductor channels and have recently become experimentally accessible due to advances in nanoscale fabrication \cite{Hirayama1}. The number of bound states increases with the window length, while their energies decrease monotonically \cite{Borisov10}. Analogous results have been obtained in three dimensions for Dirichlet waveguides with circular Neumann windows \cite{Najar1}, where the dependence of the discrete spectrum on the window radius was established.

More complex geometries lead to even richer spectral behavior. In \cite{najarolend}, a straight three-dimensional waveguide with two concentric Neumann windows located on opposite walls was considered. Compared to the single-window case, the additional geometric parameter allows for greater control over spectral and localization properties. In particular, the spectrum exhibits pronounced anticrossing phenomena as the inner window radius varies, accompanied by sharp changes in the spatial localization of eigenfunctions.

\subsection{The result}
We investigate the spectral properties of a self-adjoint elliptic operator in a spatial waveguide coupled through an elliptic window. This work is a natural continuation of the analysis performed in \cite{Najar1, najarolend}, where the coupling was realized via a circular aperture. Replacing the circular window by an elliptic one leads to a more realistic geometric model and introduces a class of anisotropic perturbations that are both experimentally relevant and mathematically nontrivial.

From a physical viewpoint, elliptic windows provide a more faithful representation of coupling devices between waveguides than circular or rectangular apertures. In practical applications, apertures produced by machining, engraving, or drilling are typically elongated and often closely resemble ellipses. The elliptic geometry thus enables controlled anisotropy of the coupling: by adjusting the ratio of the semi-axes, specific propagation directions or transverse modes may be favored, directly influencing the formation of bound states, resonances, and energy transfer between waveguides.

From a mathematical perspective, the elliptic window introduces an additional geometric parameter, namely the eccentricity, giving rise to a continuous family of domain perturbations. This setting allows for a systematic investigation of the influence of the window geometry on the spectral properties of the operator, including the appearance, disappearance, and bifurcation of discrete eigenvalues below the essential spectrum threshold, their asymptotic dependence on the semi-axes, and the breaking of symmetries present in the circular case.

The loss of rotational invariance inherent to the elliptic geometry modifies the structure of transverse mode coupling, lifts degeneracies associated with higher symmetry. Varying the aspect ratio provides a means to track the creation and splitting of eigenvalues emerging from the bottom of the essential spectrum, thereby revealing spectral mechanisms that do not occur in isotropic configurations.

From a spectral point of view, the geometry of the window plays a decisive role. In the circular case, the high degree of symmetry allows separation of variables and yields an explicit spectral characterization in terms of Bessel functions, together with a precise description of eigenvalue multiplicities and threshold resonances. By contrast, the elliptic case is more subtle. Although ellipses are spectrally determined and rigid in certain restricted classes of planar domains, the general inverse spectral problem remains open.  Existing results are mainly perturbative and concern ellipses of small eccentricity \cite{Hez}. The reduced symmetry typically leads to simple eigenvalues and the loss of separability in Cartesian coordinates; although separation is possible in elliptic coordinates, it results in Mathieu-type operators and yields only implicit spectral information.

In this work, we show that the introduction of an elliptic aperture still produces a finite family of discrete eigenvalues below the bottom of the essential spectrum. This theoretical result is complemented by numerical computations illustrating the dependence of these eigenvalues on the geometric parameters of the ellipse, thereby highlighting the quantitative effects of anisotropy and symmetry breaking on the low-lying spectral structure of the coupled waveguide system. Numerical results illustrating the variation of the fundamental state with respect to several geometric parameters are presented.

The paper is organized as follows. In Section~2, we introduce the model and recall relevant known results. Section~3 is devoted to the main results and their discussion, while Section~4 presents numerical simulation.

	\section{Model and Formulation}
	\label{sec_2}
	The system we are going to study is given in figure \ref{Fig1}. We consider a Schr\"odinger particle with mass $m_p$ whose motion is confined to a pair of parallel planes separated by the width $d$. For simplicity, we assume that they are placed at $z=0$ and $z=d$. We shall denote this configuration space by $\Omega$
	$$
	\Omega=\mathbb{R}^2\times [0,d].
	$$
	
	\begin{figure}
		\centering
		\includegraphics[width=0.95\columnwidth]{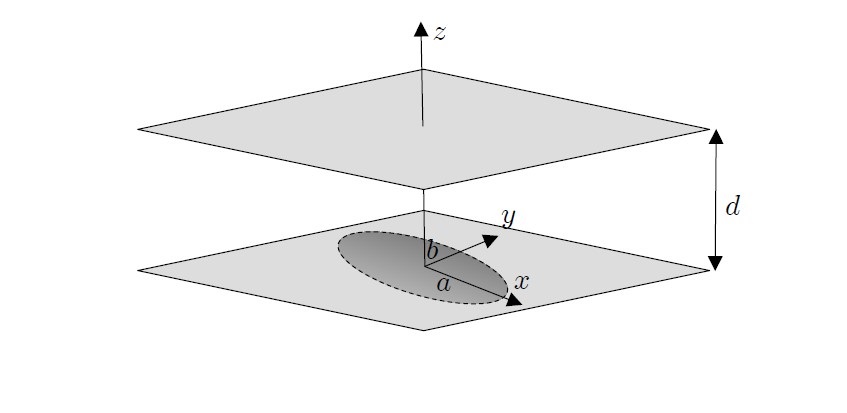}
		\caption{\label{Fig1}
			Dirichlet wave guide with elliptic Neumann window of radii $a$ and $b$.}
	\end{figure}

	Let $\gamma(a,b)$ be an ellipse of radius $a,b$ with its center at $(0,0,0)$, where
	\begin{equation}
		\gamma(a,b)=\{(x,y,0)\in \mathbb{R}^3;\ \frac{x^2}{a^2}+\frac{y^2}{b^2}=1\}.
	\end{equation}

	Without loss of generality we assume that $0< b\leq a$. We set 
	$\Gamma=\partial \Omega \backslash (\gamma(a,b))
	$. 
	We consider Dirichlet boundary condition on 
	$\Gamma$ and Neumann boundary condition in $\gamma(a,b)$  what means that $\partial\Omega_D$ and $\partial\Omega_N$ from (\ref{Dirichlet1}) and (\ref{Neumann1}) take the following form: $\partial\Omega_D\equiv\Gamma$, $\partial\Omega_N\equiv\gamma(a,b).$
	\subsection{The Hamiltonian}\label{Hamiltonian1}
	Let us define the self-adjoint operator on $L^{2}(\Omega )$ corresponding to the particle Hamiltonian $\hat{H}$. This is done by means of the quadratic forms. Namely, let $q_0$ be the quadratic form
	\begin{equation}
		 q_{0}(f,g)=\int_\Omega\overline{\nabla f}\cdot\nabla gd{\bf r},\ \mathrm{%
			with\ domain}\ \mathcal{Q}(q_{0})=\{f\in H^{1}(\Omega );\ f\lceil\partial\Omega_D=0\},
	\end{equation}
	where $H^{1}(\Omega )=\{f\in L^{2}{(\Omega )}|\nabla f\in L^{2}(\Omega )\}$ is the standard Sobolev space and we denote by $f\lceil\partial\Omega_D$ the trace of the function $f$ on $\partial\Omega_D$. It follows that $q_{0}$ is a densely defined, symmetric, positive and closed quadratic form. We denote the unique self-adjoint operator associated with $q_{0}$ by $\hat{H}$ and its domain by $D(\Omega )$. It is the Hamiltonian describing our system. From \cite{Resi} (page 276), we infer that the domain $D(\Omega)$ of $\hat{H}$ is
	$$
	D(\Omega )=\Big\{f\in H^{1}(\Omega );\ -\Delta f\in L^{2}(\Omega
	),f\lceil\partial\Omega_D=0,\frac{\partial f}{\partial z}\lceil\partial\Omega_N=0\Big\}
	$$
	and
\begin{equation*}
	\hat{H}f=-\Delta f,\ \ \forall f\in D(\Omega ),
\end{equation*}
	where we have set $\hbar^2/(2m_p)\equiv 1$.

	\subsection{Some known facts}\label{SomeKnownFacts1}
	
	Let us start this subsection by recalling that in the particular case when $a=b=0$, we get $\hat{H}^{0}$, the Dirichlet Laplacian, and at $a=b=+\infty $ we get $\hat{H}^{\infty }$, the Neumann Laplacian. Since
	$$
	\hat{H}=\left(-\Delta_{\mathbb{R}^2}\right)\otimes I\oplus I\otimes\left(-\Delta_{\lbrack 0,d]}\right),
	\mathrm{on}\ L^2(\mathbb{R}^{2})\otimes L^2([0,d]),
	$$

	(see \cite{Resi}) we get that the spectrum of $\hat{H}^0$ is $\displaystyle \left[\left(\frac{\pi}{d}\right)^{2},+\infty\right[$ and the spectrum of $\hat{H}^\infty$ is $\displaystyle [(\frac{\pi}{2d})^2,+\infty[$. Consequently, we have

	$$
	\left[\left(\frac{\pi}{d}\right)^2,+\infty\right[\subset\sigma
	\left(\hat{H}\right)\subset\left[ (\frac{\pi}{2d})^2,+\infty\right[ .
	$$
	
	Using the property that the essential spectra are preserved under compact perturbation, we deduce that the essential spectrum of $\hat{H}$ for any finite $a$ and $b$ is
	
	$$
	\sigma _{ess}\left(\hat{H}\right)=\left[(\frac{\pi}{d})^2,+\infty\right[.
	$$
	
	An immediate consequence is  that the discrete spectrum, if it exists, lies in $\displaystyle \left[(\frac{\pi}{2d})^2,(\frac{\pi}{d})^2\right]$.
\subsection{Preliminary: Elliptic coordinates system}
\label{Cylindrical1}

Due to the elliptic symmetry of the system, it is convenient to introduce an elliptic coordinate system \((r,\theta,z)\). The Cartesian coordinates \((x,y)\) are related to the elliptic coordinates \((r,\theta)\) by
\begin{equation}
	x = c \cosh r \cos \theta ,
	\label{eqn1}
\end{equation}
\begin{equation}
	y = c \sinh r \sin \theta ,
	\label{eqn2}
\end{equation}
where \(c>0\) is a geometric constant. The coordinates satisfy \(r \ge 0\) and \(0 \le \theta < 2\pi\).

Curves of constant \(r=r_0\) are obtained from
\begin{equation}
	x = c \cosh r_0 \cos \theta ,
	\label{eqn3}
\end{equation}
\begin{equation}
	y = c \sinh r_0 \sin \theta .
	\label{eqn4}
\end{equation}
Using the identity \(\cos^2\theta+\sin^2\theta=1\), these relations lead to
\begin{equation}
	\frac{x^2}{c^2 \cosh^2 r_0} + \frac{y^2}{c^2 \sinh^2 r_0} = 1 .
	\label{eqn5}
\end{equation}
Equation~\eqref{eqn5} describes an ellipse with semi-major axis
\(a=c\cosh r_0\) and semi-minor axis \(b=c\sinh r_0\), which satisfy
\begin{equation}
	a^2 - b^2 = c^2 .
	\label{eqn6}
\end{equation}
Thus, \(2c\) is the distance between the foci of the ellipse. The eccentricity is given by
\begin{equation}
	e = \frac{c}{a} = \frac{1}{\cosh r_0}
	= \sqrt{1 - \left(\frac{b}{a}\right)^2}.
	\label{eqn7}
\end{equation}
Hence, surfaces of constant \(r\) correspond to confocal ellipses, while surfaces of constant \(\theta\) are confocal hyperbolas.

Three-dimensional systems with elliptic cross-sections are conveniently described using elliptic cylindrical coordinates,
\begin{equation}
	\left\{
	\begin{array}{lll}
		x & = c\cosh r \cos \theta, \\
		y & = c\sinh r \sin \theta, \\
		z & = z .
	\end{array}
	\right.
	\label{eqn12}
\end{equation}
This coordinate system is orthogonal. The associated Laplacian operator reads
\begin{equation}
	\nabla^2 =
	\frac{2}{c^2(\cosh 2r - \cos 2\theta )}
	\left[
	\frac{\partial^2}{\partial \theta^2}
	+ \frac{\partial^2}{\partial r^2}
	\right]
	+ \frac{\partial^2}{\partial z^2}.
	\label{eqn14}
\end{equation}

\subsection{The Eigenvalue problem}\label{sec2.4}
The wavefunction is obtained by solving the time-independent Schr\"{o}dinger
equation given by
\begin{equation}
	-\frac{\hslash^2}{2m_p}\nabla^2 \Psi + V\Psi = E\Psi
	\label{eqn16a}
\end{equation}
where $V$ is the potential of the system. It is given by
\begin{equation}
	V(x,y,z) = \left\{
	\begin{array}{lc}
		0, &(x,y,z)\in \Omega \\
		\infty, & \text{otherwise}.
	\end{array}
	\right.
	\label{eqn17}
\end{equation}
The wave function inside the waveguide satisfies the time-independent Schr\"odinger equation
\begin{equation}
	-\frac{\hslash^2}{2m_p}\nabla^{2}\Psi = E\Psi .
	\label{eqn16}
\end{equation}
Exploiting the elliptic symmetry of the cross-section, we express the Laplacian in elliptic cylindrical coordinates. Equations~\eqref{eqn16a} and ~\eqref{eqn16} then takes the form
\begin{equation}
	-\frac{2}{c^{2}(\cosh 2r - \cos 2\theta)}
	\left(
	\frac{\partial^{2}\Psi}{\partial r^{2}}
	+ \frac{\partial^{2}\Psi}{\partial \theta^{2}}
	\right)
	+ \frac{\partial^{2}\Psi}{\partial z^{2}}
	= E\Psi .
	\label{eqn19}
\end{equation}

We seek separable solutions of the form
\[
\Psi(r,\theta,z) = M(r)N(\theta)Z(z).
\]
Substituting into~\eqref{eqn19} and dividing by \(MNZ\) yields
\begin{equation}
	-\frac{2}{c^{2}(\cosh 2r - \cos 2\theta)}
	\left(
	\frac{1}{M}\frac{d^{2}M}{dr^{2}}
	+ \frac{1}{N}\frac{d^{2}N}{d\theta^{2}}
	\right)
	- \frac{1}{Z}\frac{d^{2}Z}{dz^{2}}
	= E .
	\label{eqn21}
\end{equation}

The longitudinal dependence is separated by introducing the constant \(E_z\),
\begin{equation}
	-\frac{1}{Z}\frac{d^{2}Z}{dz^{2}} = E_z .
	\label{eqn22}
\end{equation}
The general solution is
\begin{equation}
	Z(z) = A\cos(\pi\sqrt{E_z}\,z) + B\sin(\pi\sqrt{E_z}\,z),
	\label{eqn23}
\end{equation}
where the constants are fixed by the boundary conditions in the \(z\)-direction.

The boundary of the elliptic window corresponds to \(r=r_0\), with
\[
r_0 = \arctanh\left(\frac{b}{a}\right),
\]
which in Cartesian coordinates is described by
\begin{equation}
	\frac{x^{2}}{a^{2}} + \frac{y^{2}}{b^{2}} = 1 .
	\label{eqn24}
\end{equation}

Substituting \(E_z\) into~\eqref{eqn21} leads to the angular and radial equations
\begin{equation}
	\frac{d^{2}N}{d\theta^{2}} + (\lambda - 2q\cos 2\theta)N = 0,
	\label{eqn27}
\end{equation}
\begin{equation}
	\frac{d^{2}M}{dr^{2}} - (\lambda - 2q\cosh 2r)M = 0,
	\label{eqn28}
\end{equation}
with
\begin{equation}
	q = \frac{c^{2}(E - E_z)}{4}.
	\label{eqn29}
\end{equation}

Equations~\eqref{eqn27} and~\eqref{eqn28} are the angular and radial Mathieu equations. Periodic angular solutions exist only for discrete characteristic values \(\lambda\), leading to Mathieu functions \(\ce_m(\theta,q)\) and \(\se_{m+1}(\theta,q)\). The associated radial solutions are given by modified Mathieu functions. Matching the solutions at the boundary \(r=r_0\) yields to the eigenvalues $E$.
Further method of computation required by the specific waveguide configuration are discussed in Section~\ref{sec4}.

	\section{Analytical results}\label{sec_3}
	Here we prove existence conditions and provide evaluations derived from the analytical consideration.
	\begin{theorem}\label{th1}
		The operator $\hat{H}$ has at least one isolated eigenvalue in $\displaystyle \left[(\frac{\pi}{2d})^2,(\frac{\pi}{d})^2\right] $ for any nonzero $a$ and $b$.
		
		Moreover, for $b$ big enough and $E(a,b)$
		being an eigenvalue of $\hat{H}$ less then $\displaystyle (\frac{\pi}{d})^2$,
		there exist positive constants $C_a$ and $C_b$ such that
		
		\begin{equation}\label{as1}
			E(a,b)\in\left(\frac{C_a}{a^2},\frac{C_b}{b^2}\right).
		\end{equation}
	\end{theorem}
	\begin{remark}
		The first claim of Theorem \ref{th1} is valid for more general shape of bounded surface $\mathcal{S}$ with Neumann boundary condition, not necessarily a disc; it suffices that the surface contains a disc of radius $a>0$.
	\end{remark}
	\textbf{Proof.} For  the proof of the the first claim one may mimic the argument given in \cite{Najar1} and adjusted for the case of the two windows; however, much simpler and elegant way is to use the fact that the Neumann window is a negative perturbation \cite{Exner6}. Thus, if the one Neumann window creates the bound state \cite{Najar1}, the insertion of the second one just pushes it lower. $\blacksquare$
	
	The proof of the second claim is based on the bracketing argument. Let us split $L^2(\Omega)$ as follows: $L^2(\Omega)=L^2(\Omega_{a,b}^{-})\oplus L^{2}(\Omega_{a,b}^{+})$, with
	\begin{eqnarray*}
		\Omega_{a,b}^{-}&=&\left\{(r,\theta ,z)\in \left[
		0,\arctanh(\frac{b}{a})\right]\times\lbrack 0,2\pi
		\lbrack\times\lbrack 0,d]\right\},\\
		\Omega_{a,b}^{+}&=&\Omega\backslash\Omega_{a,b}^{-}.
	\end{eqnarray*}
	Therefore,
	$$
	\hat{H}_{a,b}^{-,N}\oplus\hat{H}_{a,b}^{+,N}\leq\hat{H}\leq\hat{H}_{a,b}^{-,D}\oplus\hat{H}_{a,b}^{+,D}.
	$$
	Here, we index by $D$ and $N$ depending on the boundary conditions considered on the surface $\displaystyle r=\tanh ^{-1}(\frac{b}{a})$. The min-max principle leads to
	$$
	\sigma _{ess}\left(\hat{H}\right)=\sigma _{ess}\left(\hat{H}_{a,b}^{+,N}\right)=\sigma_{ess}\left(\hat{H}_{a,b}^{+,D}\right)=\left[(\frac{\pi}{d})^2,+\infty \right[.
	$$
	Hence, if $\hat{H}_{a,b}^{-,D}$ exhibits a discrete spectrum below $\displaystyle 1$, then $\hat{H}$ does as well. We mention that this is not a necessary condition. If we denote by $E_{j}\left(\hat{H}_{a,b}^{-,D}\right),E_{j}\left(\hat{H}_{a,b}^{-,N}\right)$ and $E_{j}\left(\hat{H}\right)$ the $j$-th eigenvalue of $\hat{H}_{a,b}^{-,D}$, $\hat{H}_{a,b}^{-,N}$ and $\hat{H}$, respectively, then, again, the minimax principle yields the following
	\begin{equation}
		E_{j}\left(\hat{H}_{a,b}^{-,N}\right)\leq E_{j}\left(\hat{H}\right)\leq E_{j}\left(\hat{H}_{a,b}^{-,D}\right)\label{es1}
	\end{equation}
	and for $j\geq 2$
	\begin{equation}\label{Sequence1}
		E_{j-1}\left(\hat{H}_{a,b}^{-,D}\right)\leq E_{j}\left(\hat{H}\right)\leq E_{j}\left(\hat{H}_{a,b}^{-,D}\right).
	\end{equation}
	As the computation of the eigenvalue of a frustum is not an easy task, let us remark that
	\begin{equation}\label{Sequence2}
		E_j\left(\hat{H}_{a,a}^{-,D}\right)\leq E_j\left(\hat{H}_{a,b}^{-,D}\right)\leq E_j\left(\hat{H}_{b,b}^{-,D}\right)
	\end{equation}
	(the same chain of inequalities is true for the corresponding Neumann Hamiltonians $\hat{H}^{-,N}$, too).
	Then, from equations (\ref{Sequence1}) and (\ref{Sequence2}) it follows:
	\begin{equation}\label{Sequence3}
		E_{j-1}\left(\hat{H}_{a,a}^{-,D}\right)\leq E_{j}\left(\hat{H}\right)\leq E_{j}\left(\hat{H}_{b,b}^{-,D}\right).
	\end{equation}

	The Hamiltonian $\hat{H}_{\xi,\xi}^{-,D}$ has a sequence of eigenvalues \cite{abra, Lach, wat} given by
	$$
	E_{mnl}(\xi)=\left(l\pi\right) ^2+\left(\frac{x_{|m|n}}{\xi}\right)^2,
	$$
	where $x_{|m|n}$ is the $n$-th positive zero of Bessel function of the order $|m|$ (see \cite{abra,wat}) and index $j$ amalgamates the three  quantum numbers: transverse $l$, radial $n$ and azimuthal $m$; $j\equiv(l,m,n)$. The condition
	\begin{equation}\label{gaza}
		E_{mnl}\leq (\frac{\pi}{d})^2
	\end{equation}
	yields that $l=0$, so we get $E_{mn0}(b)=\left(x_{|m|n}/b\right)^2$ and $E_{m'n'0}(a)=\left(x_{|m'|n'}/a\right)^2$ with primed $n$ and $m$ referring to the index $j-1$ in (\ref{Sequence3}). Accordingly, for any eigenvalue $\lambda(a,b)$ of the Hamiltonian $\hat{H}$, there exist $m,n,m',n' \in \mathbb{N}$ such that
	\begin{equation}\label{as2}
		\frac{x_{|m'|n'}^2}{a^2}\leq E(a,b)\leq\frac{x_{|m|n}^2}{b^2}.
	\end{equation}
	When $b$ is big enough, we get the result. $\blacksquare$
	
	The above derivation shows that the coefficients $C_a$ and $C_b$ in equation (\ref{as1}) are, actually, the squares of the zeros of the Bessel functions.

	\section{Numerical computations}	\label{sec4}

	In this section we describe the numerical implementation used to compute the discrete
	spectrum obtained from the eigenvalue problem formulated in Section~\ref{sec2.4}.
	Only those aspects relevant to the numerical procedure are discussed, avoiding
	repetition of the analytical derivations. We analyse a dependence of the eigenenergies $E$  on the values of $a$ and $b$ in the whole range of their variation.\\

	Throughout the numerical analysis, energies are measured in units of $(\pi/d)^2$,
	so that the bottom of the essential spectrum is normalized to $E=1$.
	Lengths are expressed in units of the waveguide width $d$, and momenta in units of $1/d$. \\
	
	The configuration space is divided into two regions: region~I ($r \le r_0$),
	corresponding to the elliptic window, and region~II ($r>r_0$), corresponding to the
	straight waveguide. Since the potential is symmetric about the origin, the wavefunctions
	can be divided into even and odd classes of solution. Without loss of generality, the general 
	solution of the even wavefunction is:

	\begin{align}
		\Psi^{\text{I}}_m(r,\theta,z) &= \sum_{n=0}^\infty B_n^m \Ce_m(r,q_n^\text{I})\ce_m(\theta,q_n^\text{I})Z_n^\text{I}(z),\ r\le r_0 \label{eqn4.33} \\
		\Psi^{\text{II}}_m(r,\theta,z) &= \sum_{n=0}^\infty C_n^m \Ke_m(r,|q_n^\text{II}|)\ce_m(\theta,q_n^\text{II})Z_n^\text{II}(z),\ r \ge r_0 \label{eqn4.34}. 
	\end{align}
	
	The literature has adpoted various notations to describe the modified Mathieu functions, we follow the
	notations and formulae reported in \cite{abra}. The functions $\Ce_m$ and $\Ke_m$ are radial modified Mathieu
	functions of first and third kind, respectively. The angular Mathieu function $\ce_m$ is called cosine-elliptic 
	because its nature is similar to its trignometric counterpart. Similarily, the radial Mathieu functions behave
	analogous to Bessel functions. The transverse functions $Z_n(z)$ satisfies the 1D eigenvalue problem introduced
	in Section \ref{sec2.4}, and are given by

	\begin{align}
		Z^\text{I}_n(z) & = \sqrt{2}\cos\left(\left(n+\frac{1}{2}\right)\pi z\right),\ E_n^\text{I} = \left(n+\frac{1}{2}\right)^2 \label{eqn4.35} \\
		Z^\text{II}_n(z) & = \sqrt{2}\sin((n+1)\pi z),\ E_n^\text{II} = (n+1)^2 \label{eqn4.36} 
	\end{align}
	
	where the index $n=0,1,2,\ldots$ and $E_n$ are their corresponding energies. The parameter $q_n^\text{I/II}$ depends on the
	energies of the longitudinal wavefunction and is given by 
	
	\begin{equation}
		q_n^\text{I/II} = \frac{c^2(E-E_n^\text{I/II})}{4}
		\label{eqn4.37}
	\end{equation}

	From a physical point of view,, the introduction of a window in the waveguide perturbs the particle dynamics and lead to localization in region~I. In the absence of the window, the spectrum is determined by Dirichlet boundary conditions, and the lowest admissible energy is $E_0^{\mathrm{II}} = 1$.
	
	The emergence of a bound state for a nonzero window radius can be attributed to the coupling between transverse and angular degrees of freedom induced by the Neumann window(s). This coupling shifts the lowest transverse mode below the fundamental propagation threshold, thereby producing an eigenvalue separated from the continuous spectrum. The corresponding eigenfunction is square-integrable and decays at infinity.
	
	When the window is present, the particle energy can reach to values as low as $E_0^{\mathrm{I}} = \tfrac{1}{4}$. This leads to the following condition for which a bound state can 
	exist
	\begin{equation}
		E_0^\text{I} < E <E_0^\text{II}
		\label{eqn4.38}
	\end{equation}

	with $E$ being the total energy of the particle. Mathematically, matching the wavefunctions given in equations (\ref{eqn4.33}) and (\ref{eqn4.34})
	at the boundary of the disc $r = r_0$. For a given $m$, the following system of homogeneous linear equations is obtained

	\begin{equation}
		\sum_{n=0}^\infty F_{nn'} B_n^m = 0
		\label{eqn4.39a}
	\end{equation} 

	where the matrix elements of the infinite square matrix \textbf{F} are given by

	\begin{equation}
		F_{nn'} = \left[\frac{\Ce'_m(r_0,q_n^\text{I})}{\Ce_m(r_0,q_n^\text{I})} - \frac{\Ke'_m(r_0,|q_{n'}^\text{II}|)}{\Ke_m(r_0,|q_{n'}^\text{II}|)}\right]\cdot P_{nn'}
		\label{eqn4.39}
	\end{equation}

	here the prime above the Mathieu function is its derivative with respect to its argument. The matrix element $P_{nn'}$ describes the coupling of the transverse and angular modes of different regions
	\begin{equation}
		P_{nn'} = \int_0^1Z_n^\text{I}(z)Z_{n'}^\text{II}(z)dz \cdot \int_0^{2\pi} \ce_m(\theta,q_n^\text{I})\ce_m(\theta,q_{n'}^\text{II})d\theta.
		\label{eqn4.41}
	\end{equation}

	The nontrivial solution of equation (\ref{eqn4.39a}) exists when the matrix \textbf{F} is singular, i.e., det(\textbf{F})$=0$. The root of this equation provides the eigenvalues $E$.
	Figure (\ref{fig2}) depicts the bound state energy of the ground state wavefunction $m=0$ as a function of the semi-major and semi-minor axes $a$ and $b$. When the axes lengths are 
	small, the energy of the particle is close to upper Dirichlet threshold $E_0^\text{II} = 1$. The energy decreases for large axes lengths and approaches the lower threshold $E_0^\text{I} = 1/4$.
	When $a=b$, the trend of the energy found coincides with the one in \cite{najarolend}. Furthermore, the plot is symmetric about the line $a=b$ indicating that the orientation of the elliptic window
	does not play a role in determining the energy.

	 \begin{figure}[h!]
		\centering
		\includegraphics[width=\textwidth]{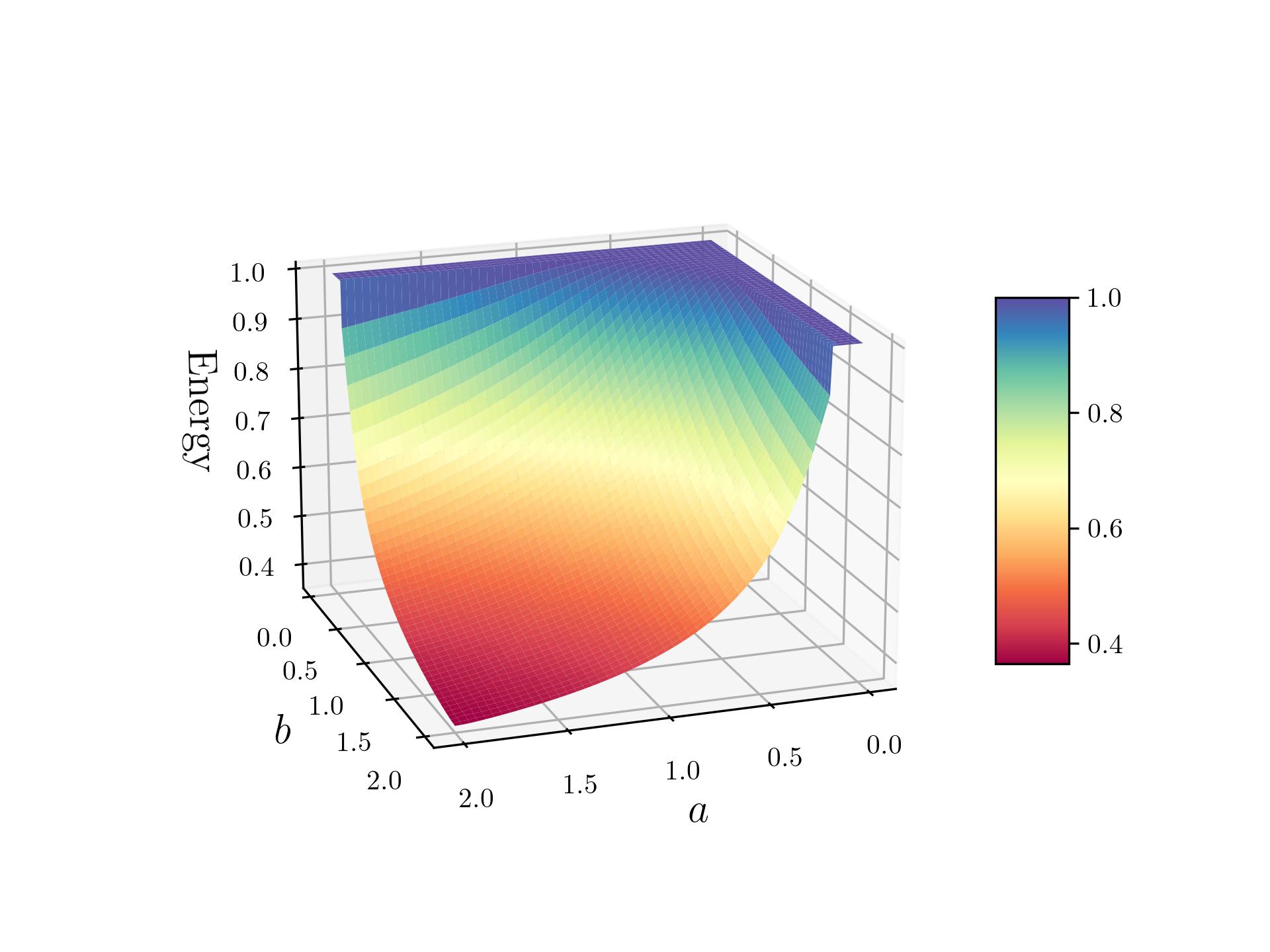}
	 	\caption{Energy spectrum of the ground state wavefunction $(m=0)$ as a function of 
	 		the elliptical axes $a$ and $b$.}
	 	\label{fig2}
	 \end{figure}
For sufficiently small radii, the energy lies very close to the fundamental
propagation threshold of the uniform Dirichlet waveguide,
$$
E^{DD}_0 = 1.
$$
It was conjectured in \cite{Exner6} that, in the case of a single window, the
energy of this bound state for small radii depends on $a$ as
$$
E_0(a) = 1 - \exp\!\left(\frac{c}{a^3}\right).
$$
In \cite{najarolend}, it was shown that the constant $c$ satisfies
$
0.44 \leq c \leq 0.45.
$
For an elliptical window that is very close to a disc, it was proved in \cite{naz} that in the
special case where $b = \varepsilon b' $  \text{and} $ a = \varepsilon a',$ the energy behaves, when $\varepsilon \to 0$, as
\begin{equation}
E_0(a,b) = 1 - \varepsilon^2 C_{a,b} + o(\varepsilon^2). \label{naz} 
\end{equation}
This quadratic behavior is confirmed by our numerical result. Indeed as shown in Figure~\ref{afixedsmall}, for $a = 0.4$ and $a = 0.5$, decreasing the eccentricity results in a clearly parabolic profile of the curve. In contrast, for $a = 1$ and $a = 1.5$, the behavior becomes hyperbolic, as illustrated in Figure~\ref{afixed}. The shape of the windows is presented in Figure~\ref{differentab}(a).\newline
As the parameter $a$ increases, the curve transitions from a concave to a convex profile. There exists a critical value $a_c$ close to $a=0.75$, which separates these two distinct behaviors and presented in Figure~\ref{acritic} .

\newpage 
	 
	 \begin{figure}[H]
		\centering
		\includegraphics[width=\textwidth]{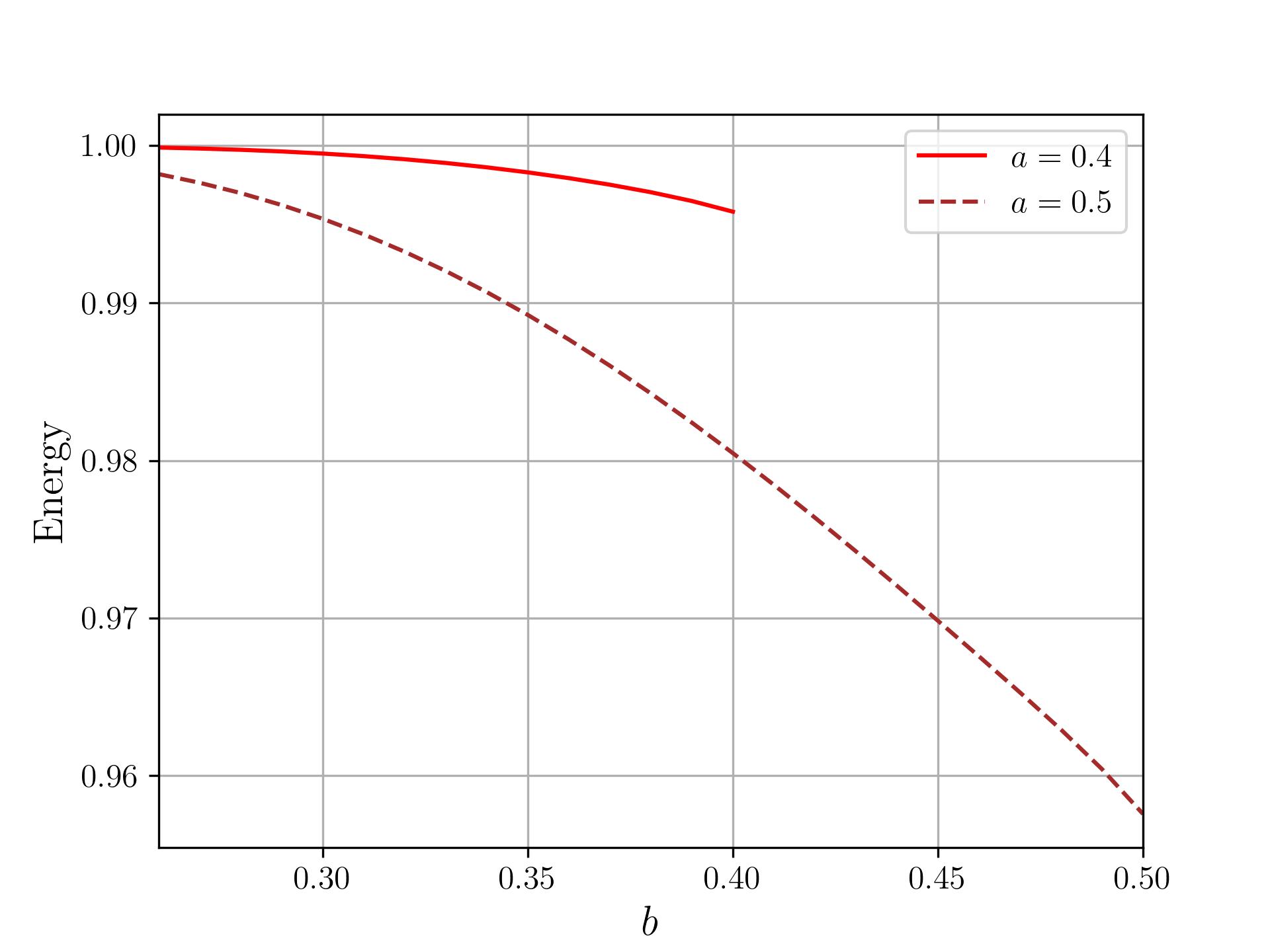}
		\caption{The energy curve when $a$ is fixed and the eccentricity is decreased}
		\label{afixedsmall}
	\end{figure}

	 \begin{figure}[H]
	 	\centering
	 	\includegraphics[width=\textwidth]{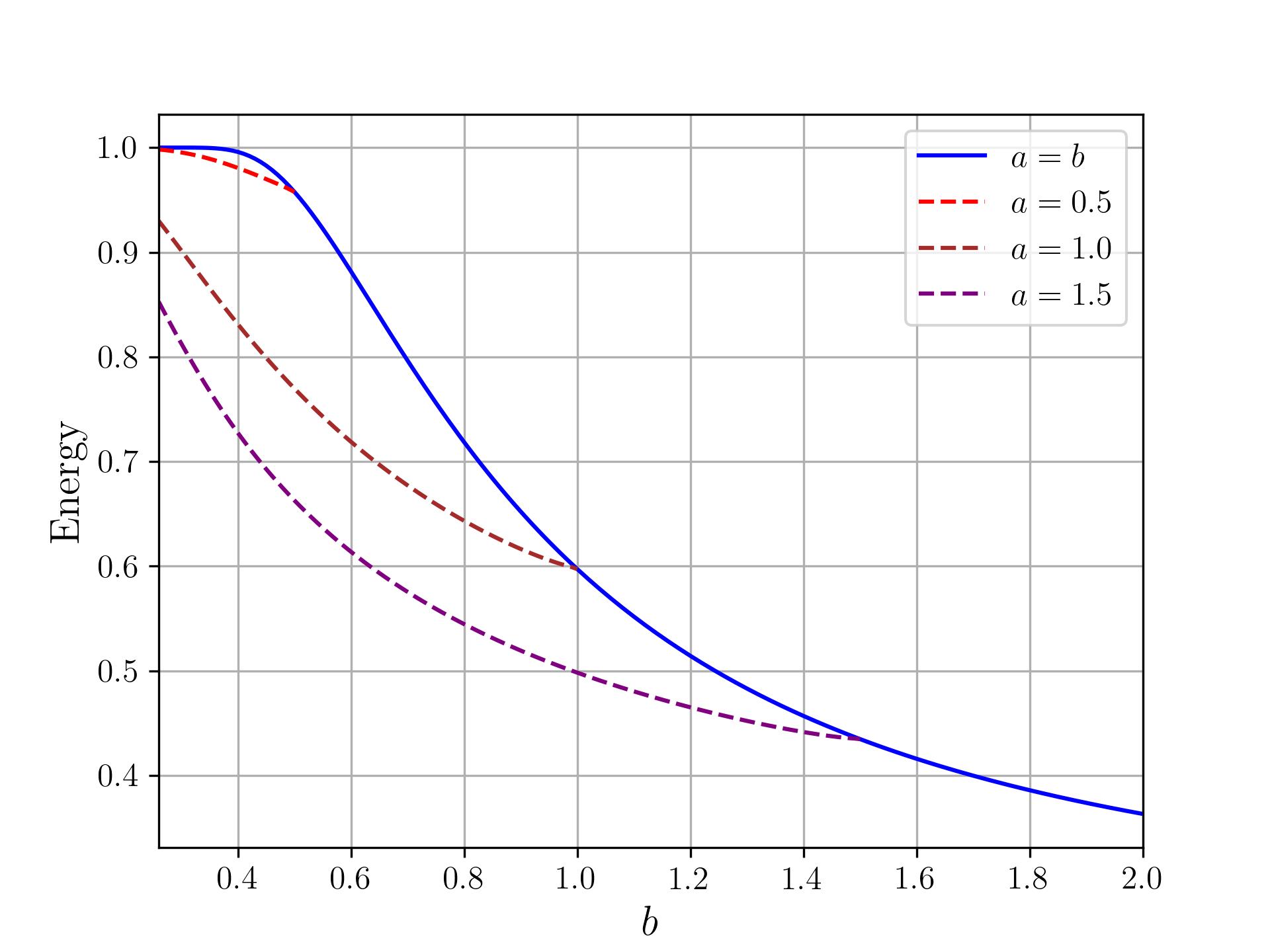}
	 	\caption{ The hyperbolic behavior of  the curves  for $a=1$ and $a=1.5$.}
	 	\label{afixed}
	 \end{figure}
	 
	 		 \begin{figure}[H]
	 	\centering
	 	\includegraphics[width=\textwidth]{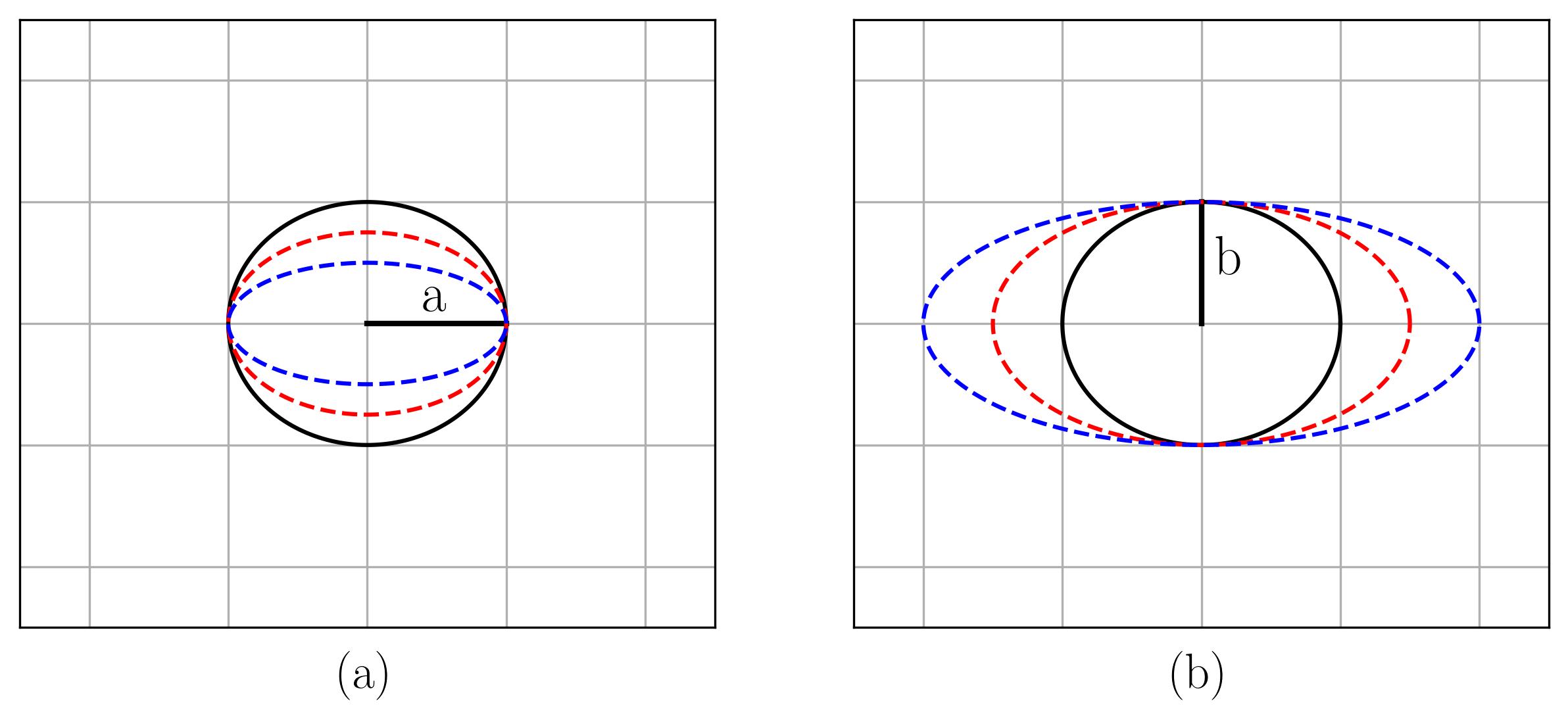}
	 	\caption{In (a), $a$ is fixed and we decrease the eccentricity (we increase $b$) of the ellipse. In (b), $b$ is fixed and we increase the eccentricity (we increase $a$) of the ellipse.}
	 	\label{differentab}
	 \end{figure}

	 \begin{figure}[H]
	 	\centering
	 	\includegraphics[width=\textwidth]{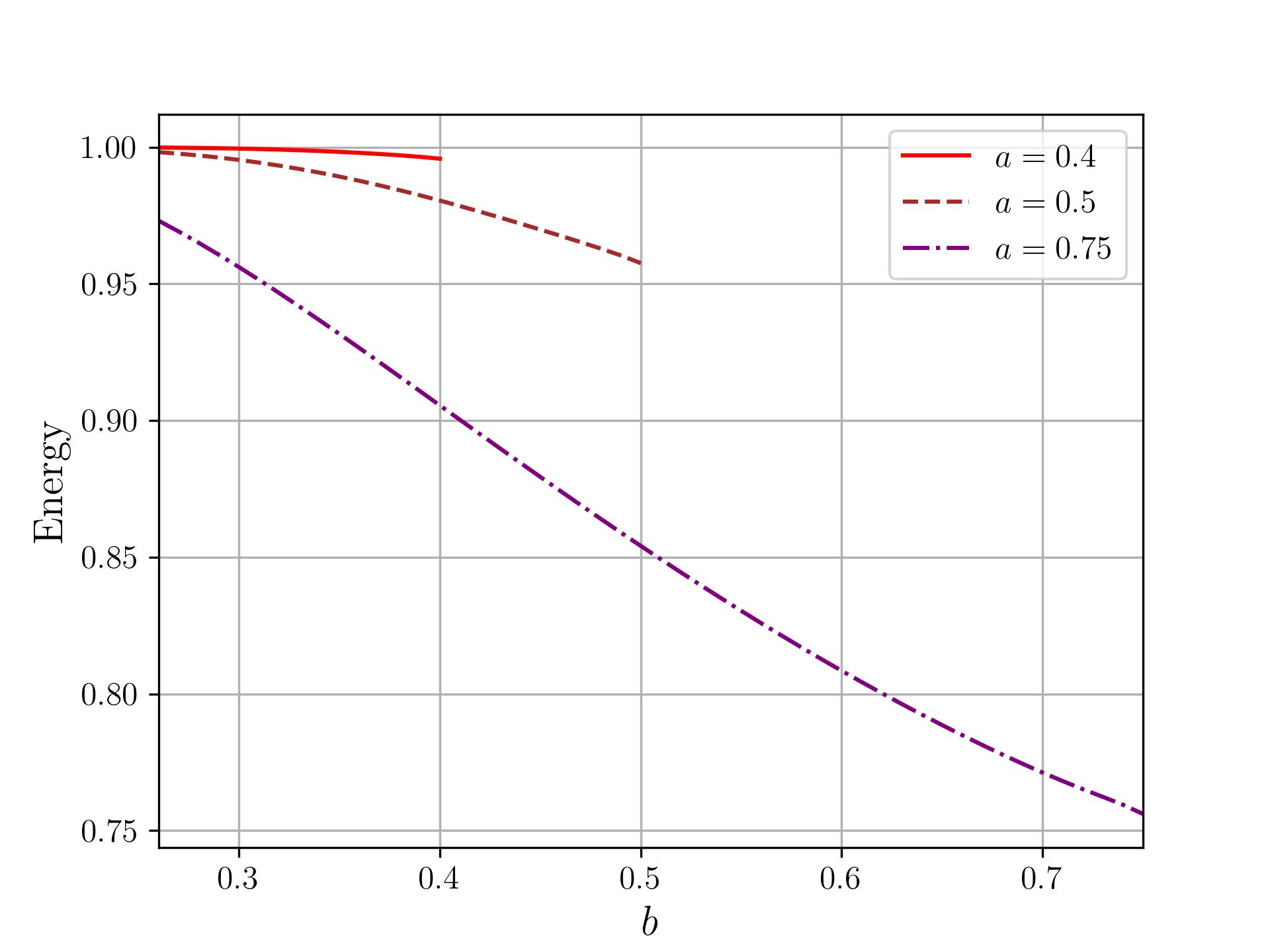}
	 	\caption{The energy curve for $a=0.75$.}
	 	\label{acritic}
	 \end{figure}
	 
When $a$ varies while $b$ is fixed with $a < b$, the corresponding energy curves are presented in Figure~\ref{bfixed}. In this regime, increasing $a$ induces an increase in the eccentricity, the shape of the window is given in Figure~\ref{differentab}(b).

In contrast with the previously analyzed configuration , no bifurcation between two distinct qualitative regimes (parabolic and hyperbolic branch ) is observed. For all considered values of $b$, the energy curves share the same qualitative behavior: an initial sharp descent followed by an asymptotically hyperbolic branch.
For $b$ small, for instance $b = 0.26$ as shown in Figure~\ref{bfixedsmall}, the energy lies very close to the fundamental
propagation threshold of the uniform Dirichlet waveguide and  the rapid decrease occurs only after a critical threshold value of $a$, approximately $ 0.5$.
From a spectral point of view, this indicates that for fixed $b$, the discrete eigenvalue emerging from the bottom of the essential spectrum decreases rapidly as $a$ increases, and subsequently converges toward a limiting value, forming a plateau. The asymptotic behavior with respect to $a$ is described by a hyperbolic branch of the form $y=C_b$
where $C_b$ denotes the asymptotic level associated with the parameter $b$. Furthermore, $C_b$ is a decreasing function of $b$, reflecting the monotonic dependence of the limiting eigenvalue on the geometric parameter. 
	 \begin{figure}[H]
	 	\centering
	 	\includegraphics[width=\textwidth]{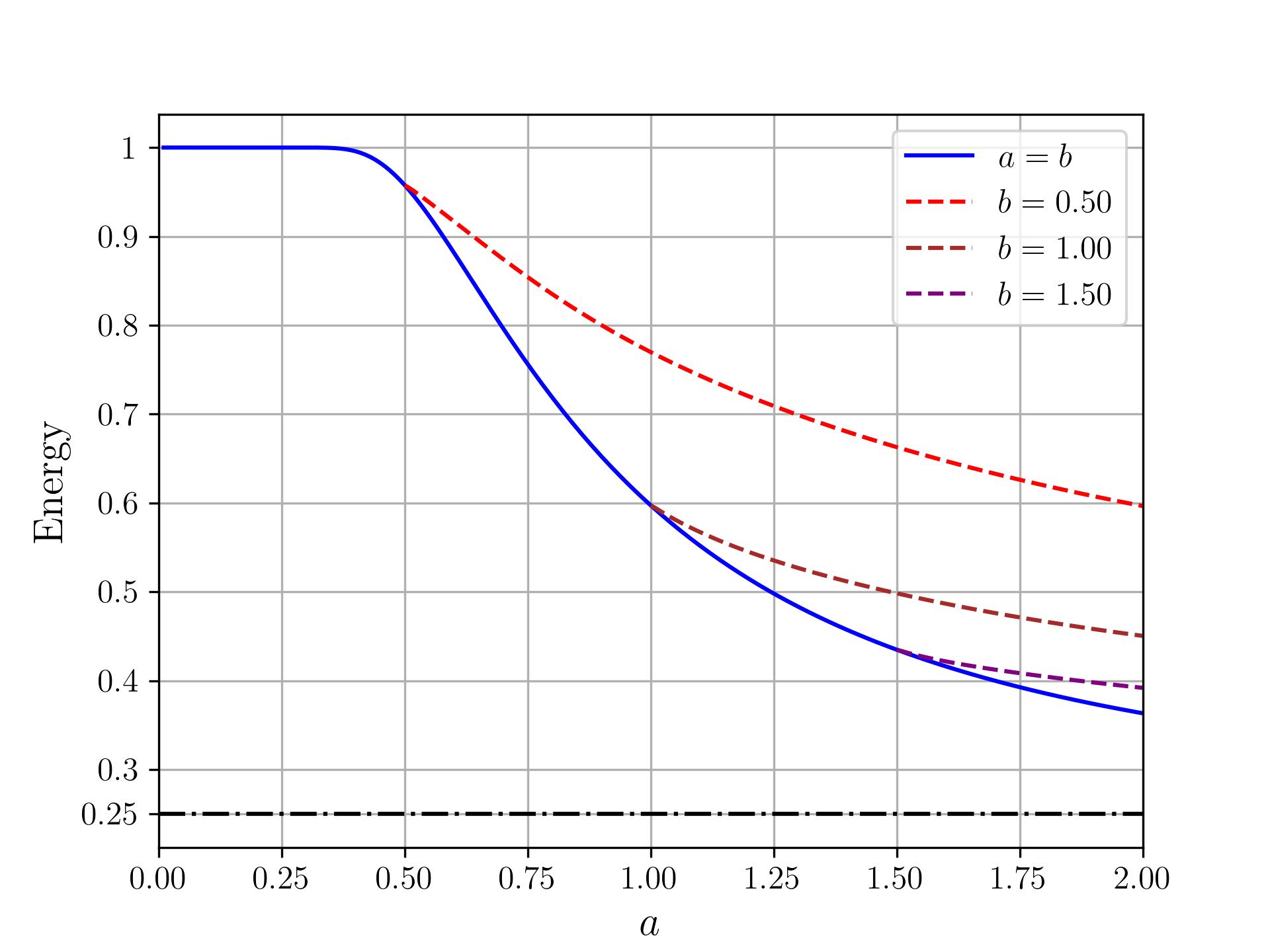}
	 	\caption{The energy curve for $b$ fixed.}
	 	\label{bfixed}
	 \end{figure}

	 \begin{figure}[H]
	 	\centering
	 	\includegraphics[width=\textwidth]{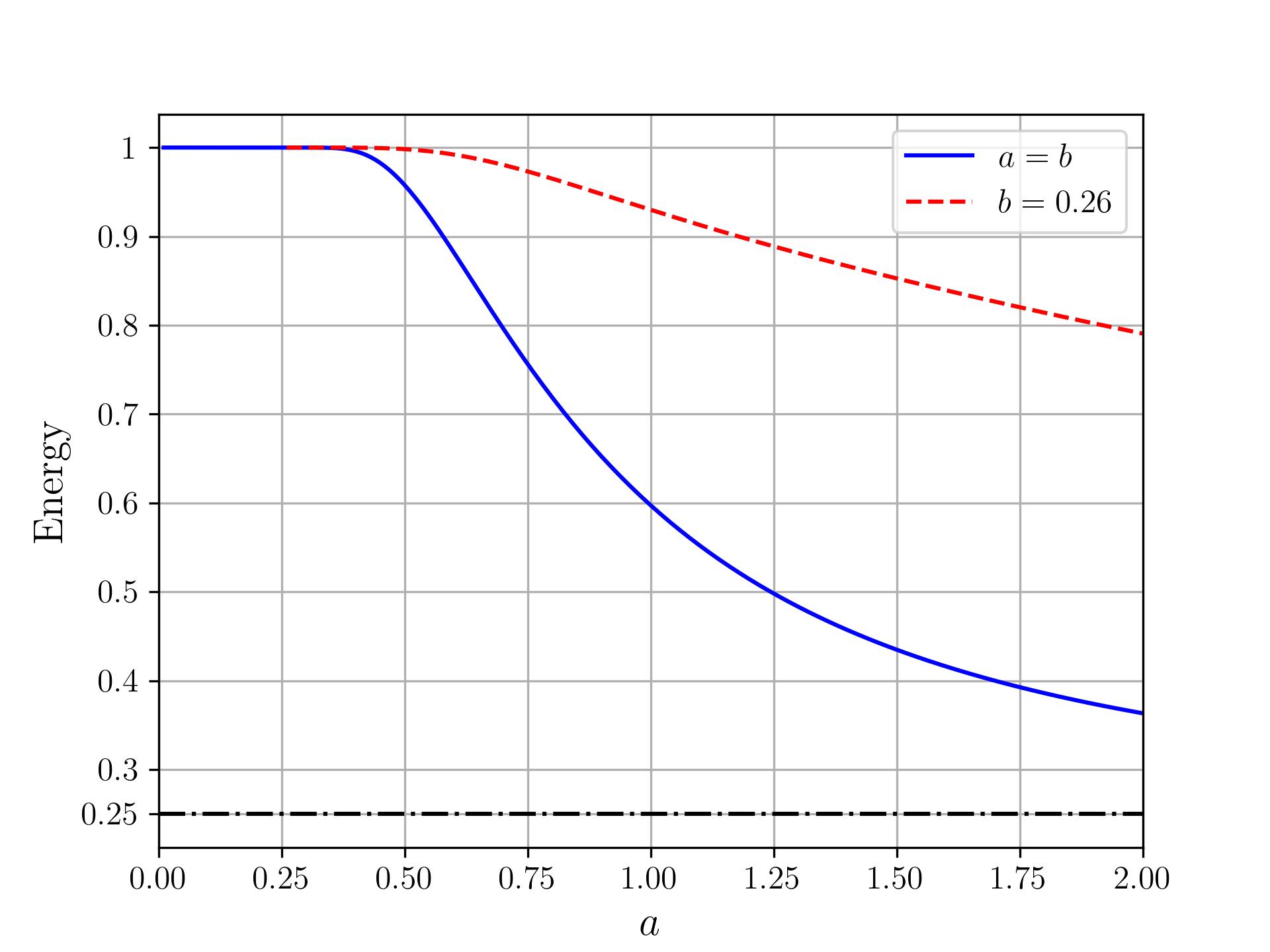}
	 	\caption{ The energy curve for $b$ small.}
	 	\label{bfixedsmall}
	 \end{figure}
	 
	 When $b$ is fixed, the energy curve as a function of the elliptical parameter $r$, is shown in Figure~\ref{bfixedr0} and the representation implies that the energy decreases as a function of the eccentricity of the window.
	 
	  \begin{figure}[H]
	 	\centering
	 	\includegraphics[width=\textwidth]{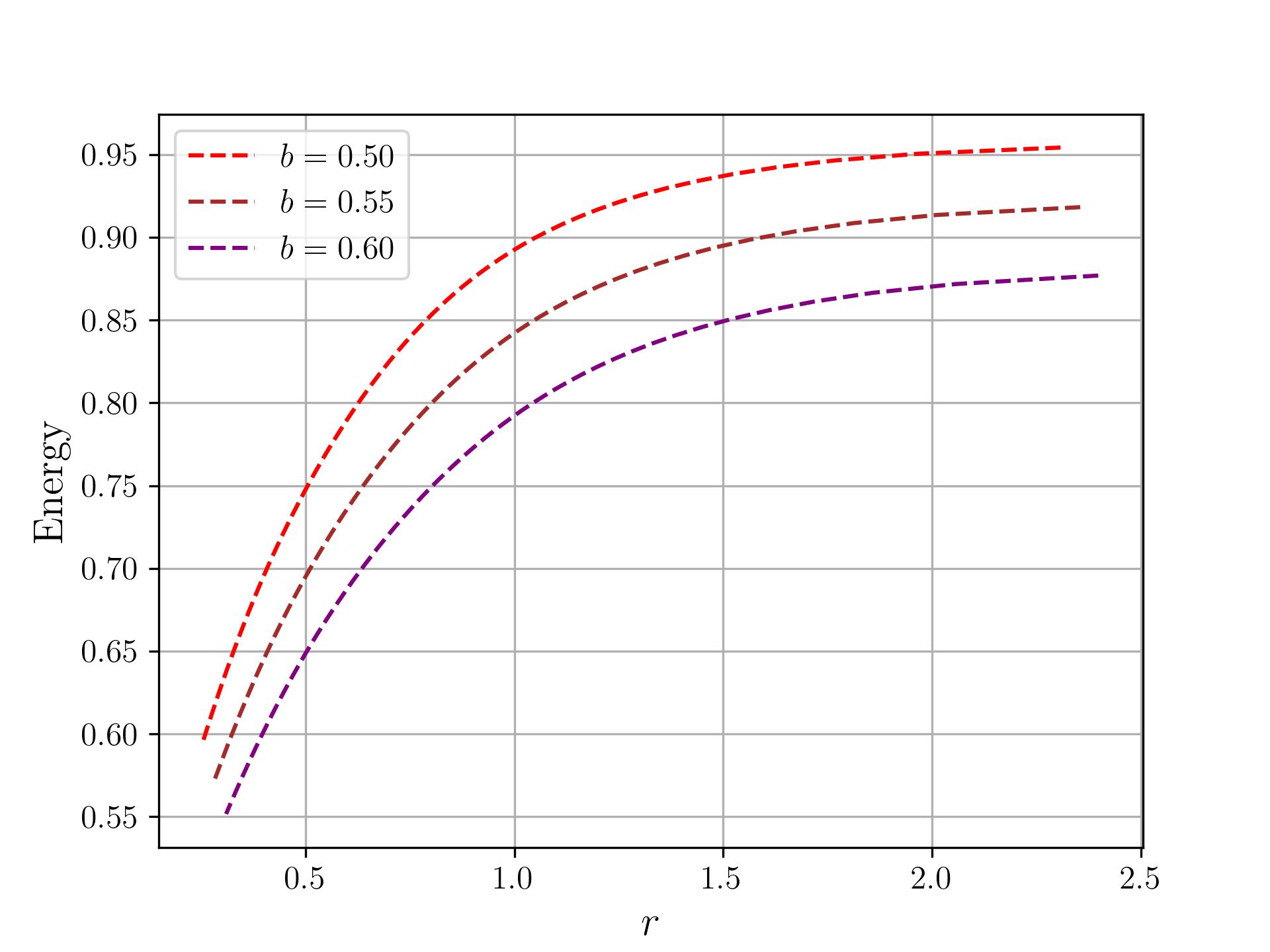}
	 	\caption{ The energy curve for $b$ fixed as a function of the elliptical coordinate.}
	 	\label{bfixedr0}
	 \end{figure}
\newpage 
\section{Conclusion}
Our analysis confirmed the presence of a bound state inside a Dirichlet waveguide with an elliptical window, irrespective of its dimensions. Its energy (eigenvalue) was shown to lie below the essential Dirichlet spectrum. The qualitative behavior of the energy curve was explained using numerical methods.

The present analysis focused on the properties of the ground-state energy. It would be instructive to study the properties of the corresponding wavefunction and higher-order states. Due to numerical challenges, this investigation is left for future work.


\begin{thebibliography}{10}
	
\bibitem[1]{abra}  M. Abramowitz and I. A. Stegun: \textsl{\textit{Handbook	of Mathematical Functions With Formulas, Graphs, and Mathematical Tables}} 	New York: Dover, (1972).
\bibitem[2]{najpop} A. S. Bagmutov, H. Najar, I. Y. Popov,  and I. F. Melikhov  \textsl{\textit{ On the discrete spectrum of a quantum waveguide with Neumann windows in presence of exterior field }} Phys. Chem. Math., 2022, {\bf{(13)}} (2), 156–163 .
\bibitem[3]{Borisov1} D. Borisov  and R. Gadyl'shin: \textsl{\textit{  On the spectrum of the Laplacian with fre- quently alternating boundary conditions}}   Teor. Mat. Fiz. {\bf (118) } 347 (in Russian) (1999).	
\bibitem[4]{Borisov2} D. Borisov: \textsl{\textit{ The asymptotics for the eigenelements of the Laplacian in a cylinder with frequently alternating boundary conditions }}  C. R. Acad. Sci. Paris Ser. II {\bf (329)} p 717-721 (2001).		
\bibitem[5]{Borisov3} D. Borisov : \textsl{\textit{On a model boundary value problem for Laplacian with frequently alternating type of boundary condition}} 	Asymptotic Analysis {\bf (35)} 1 (2003).
\bibitem[6]{Borisov4} D. Borisov: \textsl{\textit{ Asymptotics and estimates for the eigenelements of the Laplacian with frequently alternating non-periodic boundary conditions}}  Izv. RAN Ser. Mat. {\bf(67) }(6) 23 (in Russian) (2003).
\bibitem[7]{Borisov10} D. Borisov, P. Exner andR.  Gadyl'shin: \textsl{\textit{Geometric coupling thresholds in a two-dimensional strip }} J. Math. Phys. {\bf{(43)}} 2. p 6265-6278 (2002).
\bibitem[8]{exner1}  D. Borisov and P. Exner: \textsl{\textit{Exponential splitting of bound states in a waveguide with a pair of distant windows.}} J. Phys. A 37 n$^{\circ }$ 10, p3411-3428 (2004).
\bibitem[9]{Borisov13} D. Borisov D, T. Ekholm and H. Kova\v{r}{\'i}k: \textsl{\textit{ Spectrum of the magnetic Schr¨odinger operator in a waveguide with combined boundary conditions}}   Ann. Henri Poincar\'{e} {\bf (6)} 327 (2005).
\bibitem[10]{bulla1}  W. Bulla, F. Gesztesy, W. Renger, and B. Simon: \textsl{\textit{\ Weakly coupled Bound States in Quantum Waveguides.}} Proc. Amer. Math. Soc. \textbf{125} , no. 5, 1487--1495 (1997).
\bibitem[11]{Dowker1} J. S. Dowker, K. Kirsten K and P. B. Gilkey: \textsl{\textit{ On properties of asymptotic expansion of the heat trace for the N/D problem}}   Int. J. Math. {\bf (12)} p 505-517 (2001).
\bibitem[12]{Driscoll1} T. A. Driscoll and H. P. W. Gottlieb: \textsl{\textit{ Summation of perturbation series of eigenvalues and eigenfunctions of anharmonic oscillator}} Phys. Rev. E {\bf(68) } 016702 (2003).
\bibitem[13]{exner4}  P. Duclos and P. Exner: \textsl{\textit{Curvature-induced Bound States in Quantum waveguides in two and three dimensions }} Rev. Math. Phy. {(\textbf{{37})}} p 4867-4887 (1989).
\bibitem[14]{1}  P. Exner, P. \v{S}eba: \textsl{\textit{Bound states and scattering in quantum waveguides coupled laterally through a boundary window.}} J. Math. Phys. \textbf{(30)} n$^{\circ }$ 10, p 2574 (1989).
\bibitem[15]{exner2}  P. Exner, P. \v{S}eba, M. Tater, and D. Van\v{e}k: \textsl{\textit{Bound states and scattering in quantum waveguides coupled laterally through a boundary window.}} J. Math. Phys. \textbf{(37)} n$^{\circ }$ 10, p4867-4887 (1996).
\bibitem[16]{exner3}  P. Exner, S. A Vugalter: \textsl{\textit{ Asymptotic estimates for bound states in quantum waveguides coupled laterally through a narrow window.}} Ann. Inst. H. Poincare,{\bf{(65)}}, No.1, p 109–123 (1996).
\bibitem[17]{Exner6} P. Exner P and S. A.  Vugalter:  \textsl{\textit{ Bound-state asymptotic estimates for window-coupled Dirichlet strips and layers}} Jour. Phy. A {\bf(30) } 7863 (1997).			 	
\bibitem[18]{Gortinskaya1} Gortinskaya L V, Popov I Y, Tesovskaya E S and Uzdin V M 	\textsl{\textit{Electronic transport in the multilayers with very thin magnetic layers}}  Physica E {\bf (36)}  p. 12-16. (2007).
\bibitem[19]{Gre} D. S. Grebenkov, B.T. Nguyen \textsl{\textit{ Geometrical Structure of Laplacian Eigenfunctions }} SIAM REVIEW: Society for Industrial and Applied Mathematics	{\bf{(55)}}, No. 4, p. 601–667 (2013).	
\bibitem[20]{Gri}	A. Grigis and  J. Sjöstrand:  \textsl{\textit{  Microlocal Analysis for Differential Operators }} Cambridge University Press, (1994).
\bibitem[21]{Hez} H. Hezari, S. Zelditch, \textsl{\textit{Inverse spectral rigidity of elliptic domains}}
Inventiones Mathematicae {\bf(196)} , p. 1–80 (2014).
\bibitem[22]{Hirayama1}Y.  Hirayama, A. D. Wieck, T. Bever , vonK.  Klitzing and K. Ploog:  \textsl{\textit{Parallel in-plane-gated wires coupled by a ballistic window}}  Phys. Rev. B {\bf (46) } 4035 (1992).
\bibitem[23]{Holcman1} D. Holcman and Z. Schuss: \textsl{\textit{ Diffusion escape through a cluster of small absorbing windows}} Jour. Phy. A Mathematical and Theoretical {\bf (41)} 155001 (2008).
\bibitem[24]{hurt}  N. E. Hurt: \textsl{\textit{\ Mathematical Physics Of Quantum Wires and Devices}} Mathematics and its Application \textbf{(506)}Kluer Academic, Dordrecht, (2000).
\bibitem[25]{Jackson0} D. Jakobson: \textsl{\textit{ Classical Electrodynamics}} 3rd edn (New York: Wiley) sect. 3.13 (1999).
\bibitem[26]{Jakobson1} D. Jakobson, M. Levitin, N. Nadirashvili  and I.  Polterovich: \textsl{\textit{ Spectral problems with mixed Dirichlet-Neumann boundary conditions: isospectrality and beyond}}  J. Comp. Appl. Math. {\bf (194)} 141 (2006).
\bibitem[27]{Stokel}  F. Kleespies and P. Stollmann: \textsl{\textit{Lifshitz Asymptotics and Localization for random quantum waveguides.}} Rev. Math. Phy. \textbf{(12)} p 1345-1365 (2000).
\bibitem[28]{speis1}  A. Klein; J. Lacroix, and A. Speis, Athanasios: \textsl{\textit{\ Localization for the Anderson model on a strip with singular potentials.}} J. Funct. Anal. \textbf{(94)} n$^{\circ }$ 1,
p. 135-155 (1990).
\bibitem[29]{Kond} V. A. Kondrat’ev: \textsl{\textit{Boundary value problems for elliptic equations in domains with conical or angular points}} Trans. Moscow Math. Soc. {\bf{(16)}}, p. 227–313, (1967).
\bibitem[30]{Krejcirik1} D. Krej\v{c}i\v{r}\'{\i}k  and J. K\v{r}\'{\i}\v{z}: \textsl{\textit{ On the spectrum of curved quantum waveguides.}}  Publ. Res. Inst. Math. Sci  {\bf (41)}  p 757-791 (2005).	
\bibitem[31]{Levitin1} M. Levitin,  L.Parnovski and I. Polterovich : \textsl{\textit{ Isospectral domains with mixed boundary conditions}}  J. Phy. A: Math. Theo.  {\bf (39)} 2073 (2006).	
\bibitem[32]{Lach}	N.W. McLachlan: \textsl{\textit{Theory and Application of Mathieu Functions}}, Oxford Univ. Press, (1947).
\bibitem[33]{naj6}  H. Najar: \textsl{\textit{Lifshitz tails for acoustic	waves in random quantum waveguide}} Jour. Stat. Phy. Vol 128 No 4, p1093-1112 (2007).
\bibitem[34]{Najar1} H. Najar, S. Ben Hariz and M. Ben Salah:  \textsl{\textit{On the Discrete Spectrum of a Spatial Quantum Waveguide with a Disc Window}} Math. Phys. Anal. Geom. {\bf (13)} 19 (2010).	
\bibitem[35]{najarolend} H Najar and O Olendski: \textsl{\textit{Spectral and localization properties of the Dirichlet wave guide with two concentric Neumann discs }} J. Phys. A: Math. Theor. {\bf{44}} 305304 (2001).
\bibitem[36]{najrandom} H. Najar: \textsl{\textit{  Lifshitz Tails for Quantum Waveguides with Random Boundary Conditions}}.   Math. Phys. Anal. Geom. {\bf{(22)}}, 17 (2019).
\bibitem[37]{naz}  S. A. Nazarov and M. Specovius-Neugebauer: \textsl{\textit{Selfadjoint extensions of the Neumann Laplacian in domains with cylindrical outlets.}} Commu. Math. Phy. \textbf{185} p 689-707 (1997).
\bibitem[38]{Olendski2} O. Olendski and L. Mikhailovska: \textsl{\textit{  A straight quantum wave guide with mixed Dirichlet and Neumann boundary conditions in uniform magnetic fields}} Jour. Phy. A: Mathematical and Theoretical  {\bf (40)} 4609 ( 2007).
\bibitem[39]{Olendski3}O. Olendski and L. Mikhailovska: \textsl{\textit{ Analytical and numerical study of a curved planar waveguide with combined Dirichlet and Neumann boundary conditions in a uniform magnetic field}}  Phy. Rev. B,  {\bf (77)} 174405 (2008)
\bibitem[40]{Prange1} R. E. Prange,  E. Ott, T. M. Antonsen, B. Georgeot and R. Bl\"{u}mel: \textsl{\textit{Ray splitting and quantum chaos}}  Phys. Rev. E  {\bf (53)} 207 (1996).
\bibitem[41]{Resi}  M. Reed and B. Simon: \textsl{\textit{Methods of Modern	Mathematical Physics}} Vol. IV: Analysis of Operators. Academic, Press, (1978).
\bibitem[42]{Seeley1} R. Seeley:  \textsl{\textit{ Trace expansions for the Zaremba problem}}  Commun. Partial Differ. Equ. {\bf (27)} 2403 (2002).
\bibitem[43]{Wiersig1}J.  Wiersig: \textsl{\textit{ Spectral properties of quantized barrier billiards}} 2002 Physical Review E{\bf (65)} 046217 (2002).
\bibitem[44]{wat}  G. N. Watson: \textsl{\textit{A Treatise On The Theory of Bessel Functions}} Cambridge University Press.
\bibitem[45]{Zaremba1} S. Zaremba: \textsl{\textit{Sur un problème mixte relatif à l'équation de Laplace}} Bull. Intern. Acad. Sci. Cracovie 314 (1910).
\end{thebibliography}
\end{document}